\documentclass{article}

\usepackage{arxiv}

\usepackage[utf8]{inputenc} 
\usepackage[T1]{fontenc}    
\usepackage{hyperref}       
\usepackage{url}            
\usepackage{booktabs}       
\usepackage{amsfonts}       
\usepackage{nicefrac}       
\usepackage{microtype}      

\usepackage{lipsum}         
\usepackage{graphicx}
\usepackage[square,sort,numbers]{natbib}
\usepackage{doi}

\usepackage{subcaption}
\usepackage{graphicx}
\usepackage{hyperref}
\usepackage{amsmath}
\usepackage{amssymb}
\usepackage{amsfonts}
\usepackage{amsthm}
\usepackage{mathtools}
\usepackage{bm}
\usepackage{epsdice}
\usepackage{algorithm}
\usepackage{algorithmicx}
\usepackage{algpseudocode}
\usepackage{bbding}
\usepackage{xcolor}
\usepackage{sgame}

\newcommand{\nashhighlight}{\colorbox{lime}}

\newtheorem{theorem}{Theorem}
\newtheorem{lemma}[theorem]{Lemma}
\newtheorem{corollary}[theorem]{Corollary}

\theoremstyle{definition}%
\newtheorem{example}{Example}%

\theoremstyle{definition}%
\newtheorem{definition}{Definition}%

\title{On Nash Equilibria in Normal-Form Games With Vectorial Payoffs\thanks{Part of this work was carried out by the first author for his thesis \cite{ropke2021thesis} under the supervision of the other authors. Some preliminary results in this article were presented in the Multi-Objective Decision Making Workshop 2021 \cite{ropke2021nash}}}

\author{ \href{https://orcid.org/0000-0001-5045-6127}{\includegraphics[scale=0.06]{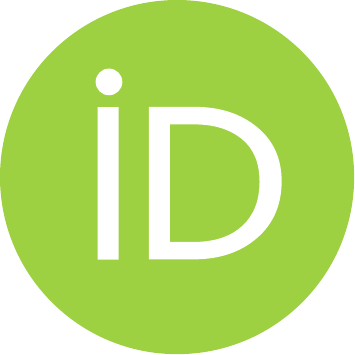}\hspace{1mm}Willem R\"{o}pke \Envelope}\\
	Artificial Intelligence Lab \\
	Vrije Universiteit Brussel, Belgium \\
	\texttt{willem.ropke@vub.be} \\
	\And
	Diederik M.\ Roijers\\
	Artificial Intelligence Lab \\
	Vrije Universiteit Brussel, Belgium \&\\
	Microsystems Technology \\
	HU~University~of~Applied~Sciences~Utrecht The~Netherlands\\
	\texttt{diederik.yamamoto-roijers@hu.nl} \\
	\And
	Ann Now\'{e} \\
	Artificial Intelligence Lab \\
	Vrije Universiteit Brussel, Belgium \\
	\texttt{ann.nowe@vub.be} \\
	\And
	Roxana R\u{a}dulescu \\
	Artificial Intelligence Lab \\
	Vrije Universiteit Brussel, Belgium \\
	\texttt{roxana.radulescu@vub.be} \\
}

\renewcommand{\shorttitle}{On Nash Equilibria in Normal-Form Games With Vectorial Payoffs}

\hypersetup{
pdftitle={On Nash Equilibria in Normal-Form Games With Vectorial Payoffs},
pdfsubject={cs.GT},
pdfauthor={Willem R\"{o}pke, Diederik M.\ Roijers, Ann Now\'{e}, Roxana R\u{a}dulescu},
pdfkeywords={Game theory, Nash equilibria, Multi-objective},
}

\begin{document}
\maketitle

\begin{abstract}
We provide an in-depth study of Nash equilibria in multi-objective normal form games (MONFGs), i.e., normal form games with vectorial payoffs. Taking a utility-based approach, we assume that each player's utility can be modelled with a utility function that maps a vector to a scalar utility. In the case of a mixed strategy, it is meaningful to apply such a scalarisation both before calculating the expectation of the payoff vector as well as after. This distinction leads to two optimisation criteria. With the first criterion, players aim to optimise the expected value of their utility function applied to the payoff vectors obtained in the game. With the second criterion, players aim to optimise the utility of expected payoff vectors given a joint strategy. Under this latter criterion, it was shown that Nash equilibria need not exist. Our first contribution is to provide a sufficient condition under which Nash equilibria are guaranteed to exist. Secondly, we show that when Nash equilibria do exist under both criteria, no equilibrium needs to be shared between the two criteria, and even the number of equilibria can differ. Thirdly, we contribute a study of pure strategy Nash equilibria under both criteria. We show that when assuming quasiconvex utility functions for players, the sets of pure strategy Nash equilibria under both optimisation criteria are equivalent. This result is further extended to games in which players adhere to different optimisation criteria. Finally, given these theoretical results, we construct an algorithm to compute all pure strategy Nash equilibria in MONFGs where players have a quasiconvex utility function.
\end{abstract}

\keywords{Game theory, Nash equilibria, Multi-objective}

\section{Introduction}
\label{sec:introduction}
Game theory is the study of interaction between rational agents \cite{leytonbrown2008essentials}. The field has influenced entire areas of research, with a notable example being economics \cite{samuelson2016game}, and shaped applications such as complex energy systems \cite{he2020application}. One of the most studied settings in game theory is the (single-objective) Normal-Form Game (NFG). In this context, Nash \cite{nash1951non} showed already in 1951 that stable outcomes, in the sense of Nash equilibria, must exist for agents in these settings. An important limitation of the traditional formulation of such games is their restriction to scalar payoffs. Throughout the years, many authors have stressed the need to extend this theory to games with vectorial payoffs, rather than only considering the scalar case \cite{shapley1959equilibrium,wierzbicki1995multiple,zapata2019maxmin}. The motivation for this comes from the fact that many real-world decision making scenarios inherently contain multiple (conflicting) criteria. As an illustration, consider the case of a commuter deciding between taking their car or bike to work with the objectives of maximising speed while also minimising fuel consumption. Such scenarios subsequently return multiple payoffs. In multi-objective games, these payoffs can be modeled as a vector where each entry in the vector corresponds to the payoff for a specific objective.

We study Nash equilibria in normal-form games with vectorial payoffs, i.e., Multi-Objective Normal-Form Games (MONFGs) \cite{blackwell1954analog}. To practically deal with vectorial rewards, we take a utility-based approach assuming that each agent has a utility function that can be used to scalarise payoff vectors \cite{hayes2022practical,roijers2017multi}. One problem that arises with this approach is that it is not immediately clear at which state of the process to apply the utility function. In NFGs we deal with payoffs of mixed strategies by calculating the expected payoff of such a strategy. In the case of vectorial payoffs, there are two different approaches, referred to as optimisation criteria \cite{roijers2013survey}, and the choice for a particular criterion depends on the preference of the player. On the one hand, we could scalarise the different possible outcomes and calculate the expectation of these utilities. This approach is also called the the Expected Scalarised Returns (ESR) criterion and is equivalent to scalarising the multi-objective game with the provided utility functions, turning it into a single-objective normal-form game. We will sometimes refer to this resulting single-objective game as the trade-off game, consistent with terminology used in other works \cite{radulescu2020utility}. Conversely, we might first take the expectation of the mixed strategy and subsequently scalarise this expected payoff vector, also referred to as the Scalarised Expected Returns (SER) criterion. Recent work showed that when players have non-linear utility functions, these two criteria are not equivalent and that Nash equilibria need not exist under SER \cite{radulescu2020utility}. Because of the discrepancy between the two optimisation criteria and the fact that existence of Nash equilibria is not guaranteed for the latter, it is in general not appropriate to simply scalarise such vectorial payoffs a priori to single objective NFGs which would allow to apply traditional game theoretical concepts and techniques. We note that in principle, we could also employ a utility function agnostic approach where no knowledge of the utility function is assumed. However, this implies that we consider expected payoff vectors in the case of mixed strategies, leading us to the SER criterion where the final scalarisation is simply unknown.

We focus on five distinct aspects of MONFGs. First, we aim to provide guarantees for the existence of a Nash equilibrium under SER by imposing sufficient restrictions on the type of utility functions that can be used. Our second goal is to study the relationship between the two optimisation criteria, SER and ESR, when both have Nash equilibria. This culminates in the conclusion that under non-linear utility functions, the number of Nash equilibria under both criteria need not be equal and no equilibria must be shared. Our third contribution restrains the Nash equilibria in question to pure strategy Nash equilibria, to find equivalences between an MONFG with known utility functions and its trade-off game. Next, we subsequently extend these results to games where some agents are optimising for SER and others for ESR. We refer to this as a blended setting and formally define Nash equilibria in such games. Lastly, we utilise the theoretical results from our study in an algorithm that is able to calculate the pure strategy Nash equilibria in a given MONFG with quasiconvex utility functions. Concretely, we contribute the following:

 \begin{enumerate}
    \item We prove the existence of a Nash equilibrium in MONFGs under the SER criterion when all agents have continuous quasiconcave utility functions.
    \item We show that assuming only strictly convex utility functions is not a sufficient guarantee for Nash equilibria to exist under SER.
    \item We show that even when Nash equilibria exist under both criteria, i.e. ESR and SER, the number of Nash equilibria need not be equal and that no equilibria must be shared.
    \item We prove that pure strategy Nash equilibria under SER must also be Nash equilibria under ESR, regardless of the utility functions the agents have.
    \item We prove that if all agents have quasiconvex utility functions, the pure strategy Nash equilibria are equivalent under SER and ESR. 
    \item We define Nash equilibria in blended settings where some agents are optimising for the SER criterion while others for ESR. We subsequently show that in such settings, pure strategy Nash equilibria can be retrieved from the trade-off game when only quasiconvex utility functions are considered.
    \item We construct an algorithm that computes a subset or all of the pure strategy Nash equilibria in a given MONFG with quasiconvex utility functions.
 \end{enumerate}

\section{Background}
\label{sec:background}
In this section, we introduce the necessary game-theoretical and mathematical background to ensure that this work is self-contained. We further provide an overview of the utility-based approach which we follow to arrive at our novel theorems.

\subsection{Normal-Form Games}
\label{sec:nfg}
Normal-Form Games (NFGs) present a concise approach for reasoning about stateless $n$-player interactions. In such games, the payoff for every player depends on the joint action that is selected. Formally \cite{leytonbrown2008essentials}:

\begin{definition}[Normal-Form Game]
\label{def:nfg}
A (finite, $n$-player) normal-form game is a tuple $(N, \mathcal{A}, p)$, where: 
\begin{itemize}
    \item $N$ is a finite set of $n$ players, indexed by $i$;
    \item $\mathcal{A} = A_1 \times \dots \times A_n$, where $A_i$ is a finite set of actions available to player $i$. Each vector $a = (a_1, \dots , a_n) \in \mathcal{A}$ is called an action profile;
    \item $p = (p_1, \dots , p_n)$ where $p_i : \mathcal{A} \to \mathbb{R}$ is a real-valued payoff function for player $i$, given an action profile.
\end{itemize}
\end{definition}



We can represent a 2-player NFG as a matrix where the payoff for each joint action is shown in the associated cell. Figure \ref{fig:prisoners-dilemma} shows a well-known example of such a matrix game, namely the prisoner's dilemma. As an illustration for the payoff mechanism, if the row player in this game opts to defect while the column player chooses to cooperate, the row player would receive a payoff of $0$ while the column player receives a payoff of $-3$.

\begin{figure}[h]
    \centering
    \begin{game}{2}{2}
                  & Cooperate        & Defect\\
        Cooperate       & $-1, -1$       & $-3, 0$\\
        Defect       & $0, -3$       & $-2, -2$
    \end{game}
    \caption{A matrix representation of the prisoner's dilemma as a normal-form game. Each cell holds the payoff for both players under the corresponding action profile.}
    \label{fig:prisoners-dilemma}
\end{figure}

In general, players are not restricted to only playing a single action, also known as a \emph{pure strategy}. Rather, they can introduce randomness into their strategy by playing a mix of actions according to some probability distribution, known as a \emph{mixed strategy}. We formally define this as follows:
\begin{definition}[Mixed strategy]
\label{def:mixed-strategy}
Let $(N, \mathcal{A}, p)$ be a normal-form game, and for any set $X$ let $\Pi(X)$ be the set of all probability distributions over $X$. Then the set of mixed strategies for player $i$ is $S_i = \Pi(A_i )$.
\end{definition}
Note that a pure strategy is a special case of a mixed strategy where one action is selected with probability 1. When dealing with mixed strategies, the associated payoffs are not directly clear from the payoff functions. To remedy this, we utilise the concept of \emph{expected payoff}. Informally, the expected payoff is the average payoff that would be obtained when playing the mixed strategy a large number of times.
\begin{definition}[Expected payoff of a mixed strategy]
\label{def:exp-payoff}
Given a normal-form game $(N, \mathcal{A}, p)$, the expected payoff $p_i$ for player $i$ of the mixed strategy profile $s = (s_1, \dots , s_n)$ is defined as
\[
p_i(s) = \sum_{a \in \mathcal{A}}p_i(a)\prod_{j = 1}^n s_j(a_j).
\]
\end{definition}

\subsection{Nash Equilibria}
We assume players to be rational decision makers, implying that they are aiming to optimise their expected payoff. Because a player's expected payoff depends on the strategy of other players as well, we can not optimise a strategy in isolation. To overcome this, we define interesting groups of outcomes called solution concepts. A fundamental solution concept that we use in this work is the Nash equilibrium (NE) \cite{nash1951non}. Intuitively, one can understand a Nash equilibrium as a joint strategy from which no player can unilaterally deviate while still improving their expected payoff.

To capture the incentive for deviating, we must first define the concept of best responses. A best response is a mixed strategy that maximises a player's expected payoff given all other players' strategies. Note that such a best response need not be unique. For the purpose of notation, we define $s_{-i} = (s_1, \cdots, s_{i-1}, s_{i+1}, \cdots, s_n)$ as the strategy profile $s$ without the strategy of player $i$, so that we may write $s = (s_i, s_{-i})$. Formally:
\begin{definition}[Best Response]
Player $i$’s best response to the strategy profile $s_{-i}$ is a mixed strategy $s^\ast_i \in S_i$ such that $p_i (s^\ast_i, s_{-i}) \geq p_i (s_i, s_{-i})$ for all strategies $s_i \in S_i$.
\end{definition}
Recall that in a Nash equilibrium, no rational player wishes to deviate from the joint strategy. We can thus define a Nash equilibrium as the strategy profile in which each strategy is a best response to all other strategies:
\begin{definition}[Nash Equilibrium]
A strategy profile $s = (s_1, \dots , s_n)$ is a Nash equilibrium if, for all agents $i$, $s_i$ is a best response to $s_{-i}$. 
\end{definition}
Note that if each $s_i$ in the NE is a pure strategy, the resulting NE is called a pure strategy Nash equilibrium.

We show the prisoner's dilemma example from Section \ref{sec:nfg} in Figure \ref{fig:ne-prisoners-dilemma} where the highlighted cell (Defect, Defect) represents the equilibrium. Observe that no player can improve their payoff by deviating, as playing another action can only reduce their expected payoff. 

\begin{figure}[h]
    \centering
    \begin{game}{2}{2}
                  & Cooperate        & Defect\\
        Cooperate   & $-1, -1$       & $-3, 0$\\
        Defect   & $0, -3$       & \nashhighlight{$-2, -2$}
    \end{game}
    \caption{The prisoner's dilemma where the highlighted cell represents the Nash equilibrium.}
    \label{fig:ne-prisoners-dilemma}
\end{figure}

We highlight the fact that, in general, Nash equilibria can be mixed strategies. An intuitive example of such a mixed strategy equilibrium occurs in the game of rock-paper-scissors \cite{roughgarden2016introduction}. Every strategy that is not a uniform distribution over the three options can be exploited by the opponent, leading to a single Nash equilibrium where agents play each action with probability $\frac{1}{3}$.

\subsection{Multi-Objective Normal-Form Games}
Multi-Objective Normal-Form Games (MONFGs) can intuitively be understood as the generalisation of (single-objective) NFGs to vectorial payoffs. We define this below \cite{radulescu2020utility}:
\begin{definition}[Multi-objective normal-form game]
\label{def:MONFG}
A (finite, $n$-player) multi-objective normal-form game is a tuple $(N, \mathcal{A}, \bm{p})$, with $d \geq 2$ objectives, where: 
\begin{itemize}
    \item $N$ is a finite set of $n$ players, indexed by $i$;
    \item $\mathcal{A} = A_1 \times \dots \times A_n$, where $A_i$ is a finite set of actions available to player $i$. Each vector $a = (a_1, \dots , a_n) \in \mathcal{A}$ is called an action profile;
    \item $\bm{p} = (\bm{p}_1, \dots , \bm{p}_n)$ where $\bm{p}_i : \mathcal{A} \to \mathbb{R}^d$ is the vectorial payoff function for player $i$, given an action profile.
\end{itemize}
\end{definition}

We denote player $i$'s payoff function $\bm{p}_i$ in bold to emphasize the fact that we are dealing with vectors rather than scalars. We can represent 2-player MONFGs as a matrix analogous to 2-player NFGs and show an example in Figure \ref{fig:monfg}. To illustrate, if the row player opts for action $B$ while the column player chooses $A$, the row player receives a payoff of $(1, 0)$ and the column player $(0, 1)$.

\begin{figure}[h]
    \centering
    \begin{game}{2}{2}
                  & $A$        & $B$\\
        $A$       & $(1, 1); (0, 0)$       & $(0, 1); (1, 0)$\\
        $B$       & $(1, 0); (0, 1)$       & $(0, 0); (1, 1)$
    \end{game}
    \caption{A matrix representation of a multi-objective normal-form game. Each cell holds the vectorial payoff for both players under the corresponding action profile.}
    \label{fig:monfg}
\end{figure}

The study of MONFGs, and specifically using a utility-based approach, has received much less attention than single-objective NFGs. In addition, work on MONFGs has been fragmented and different assumptions about the setting in which this model is used can lead to vastly different outcomes \cite{radulescu2020multiobjective}.

\subsection{Utility-Based Approach}
\label{sec:uba}
To deal with the vectorial payoffs in our setting, we adopt a utility-based approach \cite{hayes2022practical,roijers2017multi}. This approach assumes for each agent $i$ the existence of a utility function $u_i: \mathbb{R}^d \rightarrow \mathbb{R}$ that can map vectors to their scalar utility. An example of a utility function that is frequently used is the linear utility function which assigns a weight $w_{i,o}$ to each objective $o$ in the payoff function $\bm{p}_i$ and subsequently calculates the weighted sum over all objectives. 

\begin{equation}
    u_i\left(\bm{p}_i(s)\right) = \sum_{o \in O}w_{i,o}p_{i,o}(s),
\end{equation}

with $p_{i,o}(s)$ the payoff for player $i$ in objective $o$ and under strategy $s$. We note that in decision theory, utility functions are often assumed to be monotonically increasing \cite{roijers2017multi}. While reasonable, we do not make this assumption as none of our results rely on this property.


\subsubsection{Optimisation Criteria}
The utility-based approach has the advantage of scalar utilities that can be compared and ranked, which leads to a partial ordering. The drawback of this approach is that the utility function can be applied in two distinct ways. As an illustration, consider the problem of a commuter going to their work. This commuter cares about two objectives, namely minimising travel time ($p_1$) and optimising their level of comfort ($p_2$). On the one hand, it is possible that they want to optimise the utility they might derive from each individual trip. This leads to what is known as the Expected Scalarised Returns (ESR) criterion where we apply the utility function before taking the expectation \cite{roijers2018multiobjective,hayes2021distributional}.
\begin{eqnarray}
\label{eq:ESR}
   p_{i}(s) = \mathbb{E}\left[u_i\left(\bm{p}_i(s)\right)\right] = \sum_{a \in \mathcal{A}}u_i\left(\bm{p}_i(a)\right)\prod_{j = 1}^n s_j(a_j)
\end{eqnarray}
In Equation \ref{eq:ESR}, $p_{i}(s)$ is the scalar utility for player $i$ with utility function $u_i$ while following the joint strategy $s$. If the commuter is making this trip daily, it is also possible that they will aim to optimise the utility they can derive from the average payoff of multiple trips. This leads to the Scalarised Expected Returns (SER) where we perform the scalarisation after the expectation.
\begin{eqnarray}
\label{eq:SER}
    p_{i}(s) = u_i\left(\mathbb{E}\left[\bm{p}_i(s)\right]\right) = u_i\left(\sum_{a \in \mathcal{A}}\bm{p}_i(a)\prod_{j = 1}^n s_j(a_j)\right)
\end{eqnarray}

Under non-linear utility functions, the two criteria can provide different results and therefore the choice of criterion requires careful consideration \cite{roijers2013survey}. We illustrate this in Example \ref{exmp:esr-ser}.

\begin{example}
\label{exmp:esr-ser}
Consider again a commuter looking to optimise their travel time and level of comfort. It is conceivable that this commuter aims to be equally comfortable and fast by using the utility function,
\begin{equation}
    u(p_{1}, p_{2}) = p_{1} \cdot p_{2}.
\end{equation}
Assume they obtain a reward of $(0, 2)$ on day one and $(2, 0)$ on day two. If they optimise for the ESR criterion, this would result in the following utility,
\begin{equation}
\begin{split}
  \text{ESR } & = \frac{1}{2} \cdot u(0, 2) + \frac{1}{2} \cdot u(2, 0) \\
    & = \frac{1}{2} \cdot 0 + \frac{1}{2} \cdot 0 \\
    & = 0.  
\end{split}
\end{equation}
On the other hand, if they optimise for the utility that can be derived from several executions of the same strategy (SER), this would result in a different utility,
\begin{equation}
\begin{split}
  \text{SER } & = u\left(\frac{1}{2} \cdot 0 + \frac{1}{2} \cdot 2, \frac{1}{2} \cdot 2 + \frac{1}{2} \cdot 0\right) \\
    & = u(1, 1)\\
    & = 1. 
\end{split}
\end{equation}
\end{example}


\subsubsection{Motivation}
\label{sec:motivation}
The motivation for the ESR criterion in games is intuitive. Observe that Equation \ref{eq:ESR} is nearly identical to the expected payoff of a mixed strategy in single-objective games (see Definition \ref{def:exp-payoff}). The critical difference is that the player's utility function is applied to each payoff vector $\bm{p}_i(a)$. This reduces the multi-objective game to a single-objective game, where each player's payoff function in the single-objective game is defined as the composition between their utility function and vectorial payoff function. We often refer to the resulting single-objective game as the trade-off game, denoting the fact that the scalar payoff for each joint-action is a trade-off between the original objectives.

The motivation for SER in games rests on three pillars. First, SER is widely used in single-agent settings \cite{hayes2022practical}. As such, it makes sense to also study this criterion in multi-objective games. Second, existing solution criteria in multi-objective games where the utility functions of the players are not known beforehand, such as the Pareto-Nash equilibrium, are often based on the SER criterion \cite{ismaili2018existence,lozovanu2005multiobjective}. This solution concept denotes a strategy that results in an undominated payoff vector for both agents. If at a later time the utility function does become known, it is then applied to the expected payoff vector.

Finally, SER is in some cases the only appropriate optimisation criterion. We already illustrated that ESR and SER are not equivalent in general. As such, when a player genuinely strives to optimise the utility from expected returns, SER is the sole option. To motivate the use of these optimisation criteria, we present an example which features both criteria in one scenario.

\begin{example}
\label{exmp:blended}
\begin{figure}[h]
    \begin{subfigure}[t]{0.45\linewidth}
            \centering
            \begin{game}{2}{2}
                  & Cardio        & Lifting\\
                Cardio       & $(4, 1); (4, 1)$       & $(5, 1); (1, 4)$\\
                Lifting       & $(1, 4); (5, 1)$       & $(1, 3); (1, 3)$
            \end{game}
            \caption{The multi-objective reward vectors.}
            \label{fig:exmp-ser}
    \end{subfigure}%
    \qquad
    \begin{subfigure}[t]{0.45\linewidth}
            \centering
            \begin{game}{2}{2}
                  & Cardio        & Lifting\\
                Cardio       & \nashhighlight{$17; 4$}       & \nashhighlight{$26; 4$}\\
                Lifting       & $5; 5$       & $4; 3$
            \end{game}
            \caption{The ESR utilities.}
            \label{fig:exmp-esr}
    \end{subfigure}%
    \caption{An example MONFG for a blended setting. The utility functions for deriving the trade-off game in Figure \ref{fig:exmp-esr} are shown in Equation \ref{eq:blended-utilities}. The highlighted cells are pure strategy Nash equilibria under ESR.}
    \label{fig:exmp-blended}
\end{figure}

Consider a scenario where an apartment building has a shared gym for its residents containing a treadmill and some weightlifting equipment. Out of the group of residents, one is an avid athlete while the others are amateurs that sporadically go to the gym for a quick run on the treadmill. At any given day, the athlete is playing a game against another resident in selecting which equipment to use. The objectives in this game are to improve cardiovascular health and increase strength. Because the gym is small, the effectiveness of their workout takes a hit when both players opt for the same equipment as they will have to wait a while. The multi-objective payoff vectors are shown in Figure \ref{fig:exmp-ser} and the scalarised utilities in Figure \ref{fig:exmp-esr} from the utility functions,

\begin{align}
\label{eq:blended-utilities}
    u_1(p_1, p_2) &= p_1^2 + p_2 & u_2(p_1, p_2) &= p_1 \cdot p_2.
\end{align}

Intuitively, the utility function for player 1 represents their quadratic preference for running on the treadmill over doing strength training, while player 2, the athlete, aims to balance the two objectives equally given their desire for a sustainable workout routine.

As the first player comes out of the pool of residents who only occasionally go to the gym, they are naturally interested in optimising for their expected utility of each individual workout (ESR). The second player however, cares about sustaining a training schedule and thus aims to optimise the utility of their average payoff (SER).

A Nash equilibrium in this blended setting occurs when no player can unilaterally deviate and improve on their individual criterion. The only joint strategy that accomplishes this is when player 1 always goes running, while player 2 mixes their strategy uniformly over running and lifting weights.

For player 1, this is trivial as running dominates weightlifting. For player 2, we can calculate their best response by maximising the univariate function 
\begin{equation}
\begin{split}
    \left(4 \cdot x + 1 \cdot \left(1 - x \right)\right) \cdot \left(1\cdot x + 4\cdot \left(1 - x\right)\right) \quad \text{ with } 0 \leq x \leq 1,
\end{split}
\end{equation}
where $x$ represents their probability of running. Note that because strategies sum to 1, we can substitute their probability of weightlifting by $1-x$. The SER maximising strategy is to go running with probability 50\% and lifting with probability 50\% for a utility of $6.25$. This strategy strikes the optimal balance between cardio and lifting for a sustained workout routine.

It is interesting to highlight what would happen if we dismiss the SER optimisation criterion in this case. Observe that this would lead to the pure strategy Nash equilibria shown in Figure \ref{fig:exmp-esr}. It is clear that either always running or always lifting is not an adequate trade-off for the athlete as it hinders their progress and can even result in injuries. Moreover, the best outcome for the athlete under ESR, with a utility of five, occurs when player 1 goes lifting while they go running. However, this is clearly suboptimal as there is no balance in the training schedule. Under SER however, they obtain a higher utility of $6.25$ when simply playing the Nash equilibrium.
\end{example}
For additional motivating examples prescribing a specific optimisation criterion, we refer to R\u{a}dulescu \cite{radulescu2021decision}, R\u{a}dulescu et al. \cite{radulescu2020multiobjective} and Roijers et al. \cite{roijers2018multiobjective}.

\subsection{Nash Equilibria in MONFGs}
The introduction of utility functions in Section \ref{sec:uba} allows us to frame Nash equilibria in MONFGs in terms of their utility \cite{radulescu2020utility}. We first define an NE under the ESR criterion:

\begin{definition}[Nash equilibrium for expected scalarised returns]
\label{def:MOMA-NE-ESR}
A joint strategy $s^{NE}$ is a Nash equilibrium in an MONFG under the expected scalarised returns criterion if for all players $i \in N$ and all alternative strategies $s_i \in S_i$:
\[
\mathbb{E}\left[u_i\left(\bm{p}_i\left(s_{i}^{NE}, s_{-i}^{NE}\right) \right)\right] \geq  \mathbb{E}\left[u_i\left(\bm{p}_i\left(s_i, s_{-i}^{NE}\right) \right)\right]
\label{eqn:ne_esr_gt}
\]
\noindent i.e. $s^{NE}$ is a Nash equilibrium under ESR if no player can increase the \emph{expected utility of its payoffs} by deviating unilaterally from $s^{NE}$.
\end{definition}

\noindent We provide an analogous definition for the SER criterion:

\begin{definition}[Nash equilibrium for scalarised expected returns]
\label{def:MOMA-NE-SER}
A joint strategy $s^{NE}$ is a Nash equilibrium in an MONFG under the scalarised expected returns criterion if for all players $i \in N$ and all alternative strategies $s_i \in S_i$:
\[
u_i\left(\mathbb{E}\left[ \bm{p}_i \left(s^{NE}_i, s^{NE}_{-i} \right)\right]\right) \geq u_i\left(\mathbb{E} \left[\bm{p}_i \left(s_i, s^{NE}_{-i} \right)\right]\right) 
\]
\noindent i.e. $s^{NE}$ is a Nash equilibrium under SER if no player can increase the \emph{utility of its expected payoffs} by deviating unilaterally from $s^{NE}$.
\end{definition}

While it has been shown that in finite NFGs a mixed strategy Nash equilibrium always exists \cite{nash1951non}, in an MONFG when optimising for the SER criterion and with non-linear utility function this guarantee does not hold \cite{radulescu2020utility}. The fact that existence is not guaranteed in MONFGs under SER can appear peculiar. In Section \ref{sec:existence-ne} however, we specify the novel perspective that any such game can be reduced to a single-objective normal-form game with an infinite set of pure strategies. This allows us to frame an alternative intuition for the non-existence result of R\u{a}dulescu et al. \cite{radulescu2020utility} and opens up interesting avenues for future work.

\subsection{Convexity and Quasiconvexity}
In this work we consider different classes of functions, namely convex and quasiconvex functions as well as concave and quasiconcave functions, due to their useful properties. We offer below the formal definition of a convex and concave function. Note that each definition can be modified to denote strict (quasi)convexity, respectively (quasi)concavity, by replacing $\leq$ with $<$, respectively $\geq$ with $>$, and considering $\lambda \in (0, 1)$.

\begin{definition}
\label{def:convex}
A function $f: \mathbb{R}^n \to \mathbb{R}$ is convex if its domain is a convex set and for all $\bm{x}_1, \bm{x}_2$ in its domain, and all $\lambda \in [0, 1]$, we have 
\begin{equation*}
    f\left(\lambda \bm{x}_1 + (1-\lambda)\bm{x}_2\right) \leq \lambda f(\bm{x}_1) + (1-\lambda)f(\bm{x}_2).
\end{equation*}
\end{definition}

\begin{definition}
\label{def:concave}
A function $f: \mathbb{R}^n \to \mathbb{R}$ is concave if its domain is a convex set and for all $\bm{x}_1, \bm{x}_2$ in its domain, and all $\lambda \in [0, 1]$, we have
\begin{equation*}
    f(\lambda \bm{x}_1 + (1-\lambda)\bm{x}_2) \geq \lambda f(\bm{x}_1) + (1-\lambda)f(\bm{x}_2).
\end{equation*}
\end{definition}

Intuitively, convex functions can be considered functions for which the line segment between any two points lies above the graph. Conversely, the line segment between any two points on a concave function will always lie below the graph. To visually illustrate these definitions, we show an example in Figure \ref{fig:convex-concave}.

\begin{figure}[h]
    \centering
    \begin{subfigure}[b]{0.45\textwidth}
        \includegraphics[width=\linewidth]{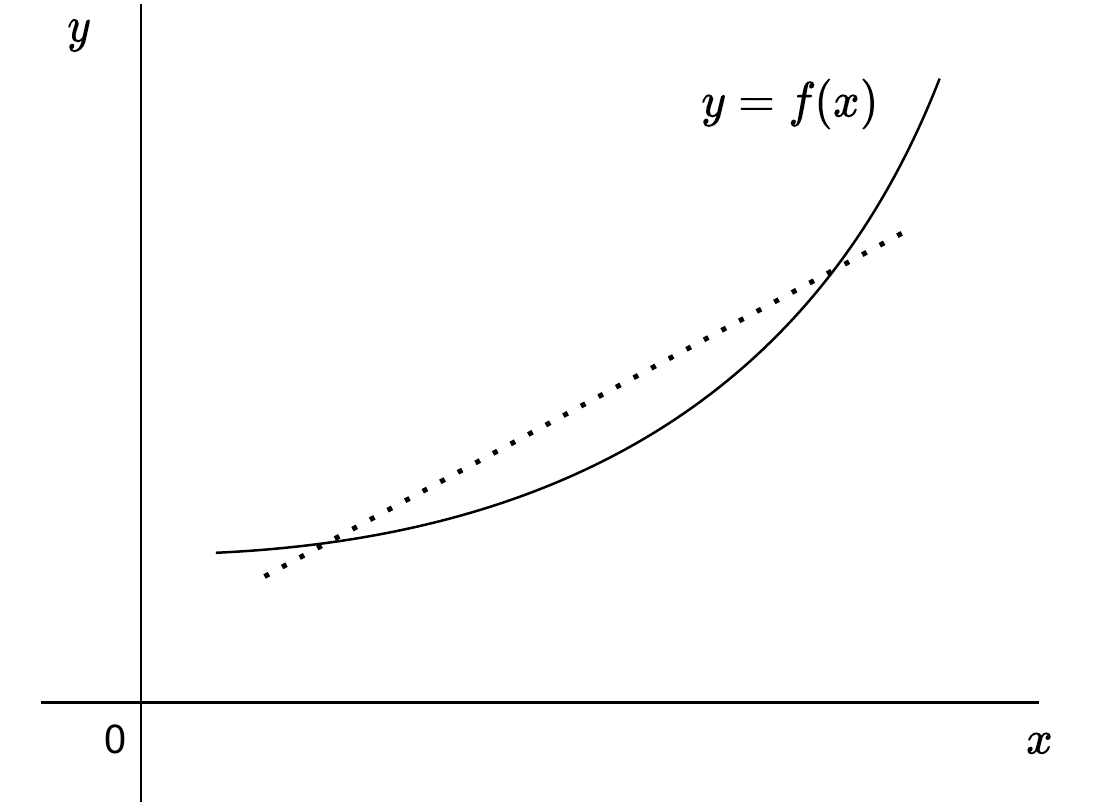}
        \caption{A convex function.}
        \label{fig:convex}
    \end{subfigure}%
    \qquad
    \begin{subfigure}[b]{0.45\textwidth}
        \includegraphics[width=\linewidth]{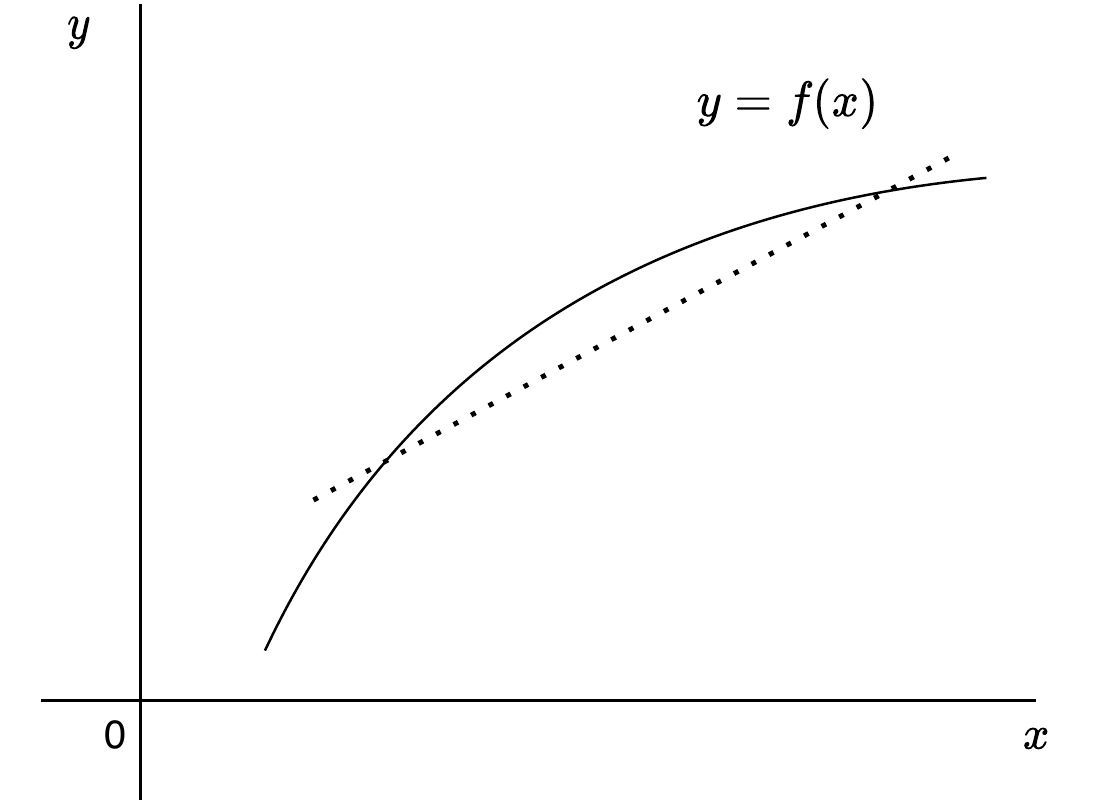}
        \caption{A concave function.}
        \label{fig:concave}
    \end{subfigure}%
    \caption{Examples of a (monotonically increasing) convex and concave function in $\mathbb{R}$. Note that the dotted lines illustrate that the line segment between any two points on such functions always lies above or below the graph between them.}
    \label{fig:convex-concave}
\end{figure}

Because the definition of such functions is relatively strict, a broader class of functions can be constructed in the form of quasiconvex and quasiconcave functions.

\begin{definition}
\label{def:quasiconvex}
A function $f: \mathbb{R}^n \to \mathbb{R}$ is quasiconvex if its domain is a convex set and for all $\bm{x}_1, \bm{x}_2$ in its domain, and all $\lambda \in [0, 1]$, we have
\begin{equation*}
    f(\lambda \bm{x}_1 + (1-\lambda)\bm{x}_2) \leq \max\left\{f(\bm{x}_1), f(\bm{x}_2)\right\}.
\end{equation*}
\end{definition}
Observe that the definition of convex functions imply that they are also quasiconvex. However, quasiconvexity does not necessarily imply convexity. We define a quasiconcave function analogously as follows:
\begin{definition}
\label{def:quasiconcave}
A function $f: \mathbb{R}^n \to \mathbb{R}$ is quasiconcave if its domain is a convex set and for all $\bm{x}_1, \bm{x}_2$ in its domain, and all $\lambda \in [0, 1]$, we have
\begin{equation*}
    f(\lambda \bm{x}_1 + (1-\lambda)\bm{x}_2) \geq \min\left\{f(\bm{x}_1), f(\bm{x}_2)\right\}.
\end{equation*}
\end{definition}
Here too, concavity implies quasiconcavity while the other way around does not hold. As such, these classes of functions can be seen as less restrictive. We show an example of a convex function, a quasiconvex function and an arbitrary function which possesses neither property in Figure \ref{fig:comparison}.

\begin{figure}[h]
    \centering
    \begin{subfigure}[t]{0.3\textwidth}
        \includegraphics[width=\linewidth]{convex.pdf}
        \caption{A convex function.}
        \label{fig:comparison-convex} 
    \end{subfigure}%
    \quad
    \begin{subfigure}[t]{0.3\textwidth}
        \includegraphics[width=\linewidth]{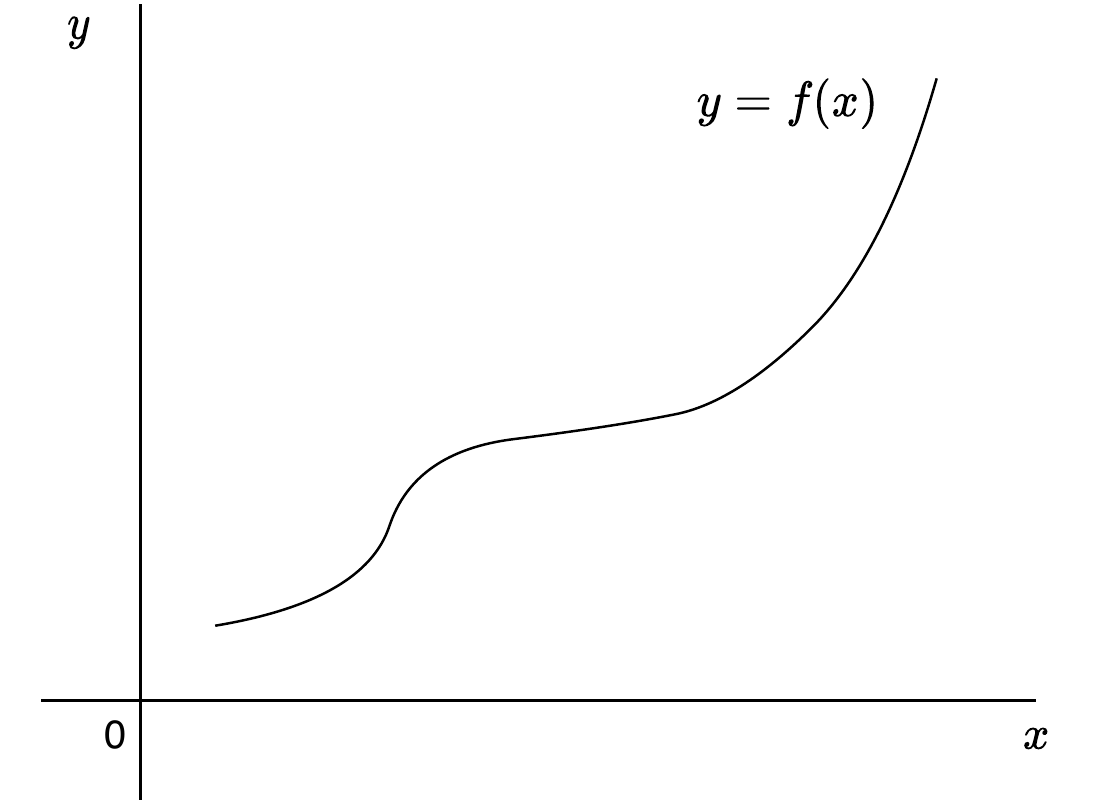}
        \caption{A quasiconvex function which is not convex.}
        \label{fig:comparison-qqoncvex}
    \end{subfigure}%
    \quad
    \begin{subfigure}[t]{0.3\textwidth}
        \includegraphics[width=\linewidth]{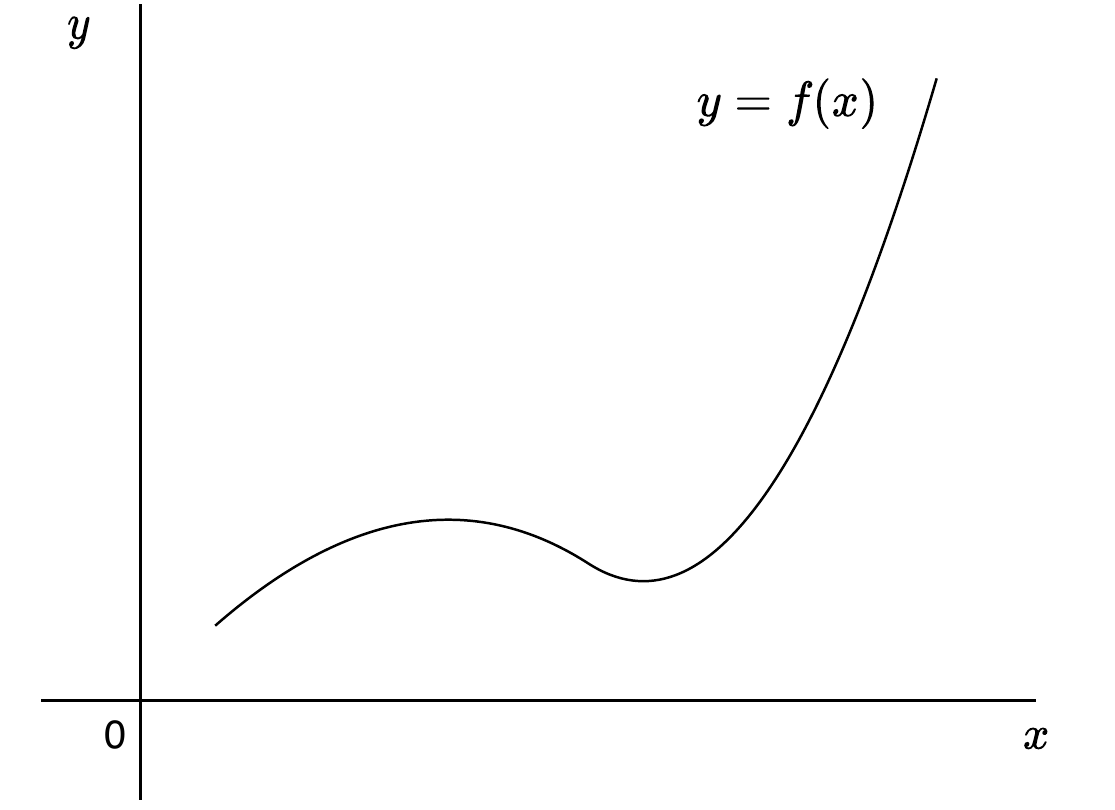}
        \caption{A function that is neither convex nor quasiconvex.}
        \label{fig:comparison-neither}
    \end{subfigure}%
    \caption{Examples of a convex function, a quasiconvex function that is not convex and an arbitrary function which is not convex nor quasiconvex in $\mathbb{R}$.}
    \label{fig:comparison}
\end{figure}

It is important to note that any function $f: \mathbb{R} \to \mathbb{R}$ that is monotonically increasing is also both quasiconvex and quasiconcave. However, this property does not generalise to multivariate functions, i.e. $f: \mathbb{R}^n \to \mathbb{R}$.

\section{Existence of Nash Equilibria in MONFGs}
\label{sec:existence}
In this section, we provide a sufficient guarantee for Nash equilibria to exist in MONFGs when optimising for the SER criterion. While it has previously been shown that no NE need exist in this setting \cite{radulescu2020utility}, we show that restricting players to use continuous quasiconcave utility functions leads to this novel result. Furthermore, we find that restricting utility functions to be strictly convex is not a sufficient guarantee which we demonstrate by means of a counterexample.

\subsection{Existence Guarantees For Continuous Quasiconcave Utility Functions}
When all players in an arbitrary finite MONFG are assumed to use a continuous quasiconcave utility function, we show that a mixed strategy Nash equilibrium is guaranteed to exist. To obtain this result, we construct an auxiliary lemma showing that every MONFG with continuous utility functions can be reduced to an equivalent continuous game.

We follow the definition of continuous games from Stein et al. \cite{stein2008separable}. Concretely, in a continuous game every player has a nonempty compact metric space of pure strategies and a continuous utility function.

\begin{lemma}
\label{lemma:continuous-reduction}
Every multi-objective normal-form game with continuous utility functions can be reduced to a continuous game.
\end{lemma}

\begin{proof} 
Let ${G = (N, \mathcal{A}, \bm{p})}$ be a finite multi-objective normal-form game and assume a continuous utility function $u_i$ is given for each player $i$. We construct a continuous game, $\hat{G} = (N, \mathcal{S}, \hat{p})$ with the same set of players $N$. 

We define player $i$'s strategy set $S_i$ in $\hat{G}$ as the set of mixed strategies over their action set $A_i$ in $G$,
\begin{equation}
    \forall i, S_i = \Pi(A_i).
\end{equation}
This ensures that every strategy set in $\hat{G}$ is a non-empty compact metric space as each set of mixed strategies is a simplex. 
 
Next, we define each player's payoff function $\hat{p}_i$ in $\hat{G}$ to be equivalent to the scalarised expected returns in $G$,
\begin{equation}
    \forall i, \hat{p}_i(s_i, s_{-i})= u_i\left(\mathbb{E}\left[\bm{p}_i\left(s_i, s_{-i}\right)\right]\right).
\end{equation} 

From Equation \ref{eq:SER} we know that each expected payoff function $\bm{p}_i$ is defined as,

\begin{equation}
    \mathbb{E}\left[ \bm{p}_i(s)\right] = \sum_{a \in \mathcal{A}}\bm{p}_i(a)\prod_{j = 1}^n s_j(a_j).
\end{equation}

Observe first that each expected payoff function is continuous and that $u_i$ was assumed to be continuous as well. It is known that a composition of continuous functions is continuous itself. As such, each $\hat{p}_i$ is continuous and $\hat{G}$ is a continuous game.
\end{proof}

This reduction is interesting in its own right and sheds a new light on earlier results obtained by R\u{a}dulescu et al. \cite{radulescu2020utility}. Specifically, it is known that continuous games need not have a Pure Strategy Nash Equilibrium (PSNE) \cite{fudenberg1991game}. Recall that our reduction from MONFGs to continuous games makes every mixed strategy in the MONFG a pure strategy in the continuous game. Informally, this implies that a PSNE in the continuous game would be a mixed-strategy NE in the MONFG. It is thus not surprising that there exist MONFGs which, when reduced, fall in the category of continuous game without PSNE and thus have no mixed-strategy NE themselves.

In continuous games, it is known that if all payoff functions are continuous and quasiconcave, a PSNE must exist \cite{fudenberg1991game}. We can leverage Lemma \ref{lemma:continuous-reduction} to show that the same restriction on utility functions in MONFGs ensures the existence of a mixed-strategy NE. 

\begin{theorem}
\label{th:ne-quasiconcave}
Consider a (finite, $n$-player) multi-objective normal-form game where players are optimising for the scalarised expected returns criterion. If each player has a continuous quasiconcave utility function, a mixed strategy Nash equilibrium is guaranteed to exist.
\end{theorem}

\begin{proof}
Let ${G = (N, \mathcal{A}, \bm{p})}$ be a finite multi-objective normal-form game and assume a continuous quasiconcave utility function $u_i$ is given for each player $i$. Lemma \ref{lemma:continuous-reduction} states that we can reduce $G$ to an equivalent continuous game $\hat{G}$ where every payoff function $\hat{p}_i(s_i, s_{-i})= u_i\left(\mathbb{E}\left[\bm{p}_i\left(s_i, s_{-i}\right)\right]\right)$. 

We can guarantee a pure strategy Nash equilibrium in $\hat{G}$ if \cite{fudenberg1991game},
\begin{enumerate}
    \item Each $S_i$ is a non-empty, compact and convex metric space;
    \item All $\hat{p}_i(s_i, s_{-i})$ are continuous in $s$;
    \item All $\hat{p}_i(s_i, s_{-i})$ are quasiconcave in $s_i$.
\end{enumerate}
Observe that (1) and (2) follow immediately from Lemma \ref{lemma:continuous-reduction} as $\hat{G}$ is a continuous game. We need only show that each $\hat{p}_i$ is quasiconcave if $u_i$ is quasiconcave.

Because $u_i$ is quasiconcave for all vectors $\bm{x}_1, \bm{x}_2$ in its domain, and all $\lambda \in [0, 1]$, we have
\begin{equation}
    u_i(\lambda \bm{x}_1 + (1-\lambda)\bm{x}_2) \geq \min\{u_i(\bm{x}_1), u_i(\bm{x}_2)\}.
\end{equation}
Let $\bm{x}_1 = \mathbb{E}\left[\bm{p}_i\left(s_1, s_{-i}\right)\right]$ and $\bm{x}_2 = \mathbb{E}\left[\bm{p}_i\left(s_2, s_{-i}\right)\right]$. Then,
\begin{equation}
    u_i\left(\lambda\mathbb{E}\left[\bm{p}_i\left(s_1, s_{-i}\right)\right] + (1-\lambda)\mathbb{E}\left[\bm{p}_i\left(s_2, s_{-i}\right)\right]\right) \geq \min\left\{u_i\left(\mathbb{E}\left[\bm{p}_i\left(s_1, s_{-i}\right)\right]\right), u_i\left(\mathbb{E}\left[\bm{p}_i\left(s_2, s_{-i}\right)\right]\right)\right\}.
\end{equation}
By the linearity of expectation we have,
\begin{equation}
    \lambda\mathbb{E}\left[\bm{p}_i\left(s_1, s_{-i}\right)\right] + (1-\lambda)\mathbb{E}\left[\bm{p}_i\left(s_2, s_{-i}\right)\right] =  \mathbb{E}\left[\bm{p}_i\left(\lambda s_1 + (1-\lambda)s_2, s_{-i}\right)\right].
\end{equation}
We may now state that,
\begin{equation}
\begin{split}
    u_i(\mathbb{E}\left[\bm{p}_i\left(\lambda s_1 + (1-\lambda)s_2, s_{-i}\right)\right]) & \geq \min\{u_i(\mathbb{E}\left[\bm{p}_i\left(s_1, s_{-i}\right)\right]), u_i(\mathbb{E}\left[\bm{p}_i\left(s_2, s_{-i}\right)\right])\} \\
    \implies \hat{p}_i\left(\lambda s_1 + (1-\lambda)s_2, s_{-i}\right) & \geq \min\{\hat{p}_i\left(s_1, s_{-i}\right), \hat{p}_i\left(s_2, s_{-i}\right)\},
\end{split}
\end{equation}
which is the requirement for quasiconcavity of $\hat{p}_i$ in $s_{i}$. As we satisfy the three conditions, we know that $\hat{G}$ must have a pure strategy Nash equilibrium. We now show that this pure strategy Nash equilibrium in $\hat{G}$ corresponds to a Nash equilibrium in $G$.

Observe that a pure strategy Nash equilibrium $s^{NE}$ in $\hat{G}$ implies that 
\begin{equation}
    \hat{p}_i\left(s_i^{NE}, s_{-i}^{NE}\right) \geq \hat{p}_i\left(s_i, s_{-i}^{NE}\right),
\end{equation}
for all players $i$ and alternative pure strategies $s_i$. Given our definition of $\hat{G}$, we can substitute $\hat{p}_i$ and arrive at the following equation:
\begin{equation}
    u_i\left(\mathbb{E}\left[ \bm{p}_i \left(s^{NE}_i, s^{NE}_{-i} \right)\right]\right) \geq u_i\left(\mathbb{E}\left[ \bm{p}_i \left(s_i, s^{NE}_{-i} \right)\right]\right), 
\end{equation}
for all players $i$ and alternative strategies $s_i$. Because each strategy $s_i \in S_i$ is defined as a mixed strategy in $G$, we obtain precisely the definition of a Nash equilibrium in an MONFG under SER. 
\end{proof}

It is worth briefly discussing what the proposed restriction intuitively means. It is known from utility theory that having a quasiconcave utility function leads to having convex preferences over the objectives. This can be understood as a player that favours an average return over their objectives more than an extreme return on just one objective. This is a sensible restriction, as quasiconcave utility functions are often considered in an economical context and are known as generating well-behaved preferences \cite{varian2014intermediate}.

\subsection{Non-Existence For Strict Quasiconvex Utility Functions}
Given the positive result, it is tempting to assume an analogous restriction to continuous quasiconvex utility functions. However, we find that this is not a sufficient restriction. In fact, the more stringent notion of strict convexity is also not a sufficient guarantee.

To show that existence is not guaranteed under these restrictions, we first introduce three necessary lemmas. Lemma \ref{lemma:optimal-ps-ms} demonstrates that when employing only quasiconvex utility functions, a pure strategy is always a best response. Lemma \ref{lemma:jensen-strict-qqconv}, Jensen's inequality for strictly quasiconvex functions, is a straightforward extension of prior results but for which we were not able to find a reference. Lastly, Lemma \ref{lemma:equal-payoff} combines the previous statements to show that whenever all utility functions are \emph{strictly} quasiconvex, a mixed strategy can only be a best response if the expected payoff vector is equal for all actions that are played with non-zero probability.

\begin{lemma}[Optimality of a pure strategy as a best response]
\label{lemma:optimal-ps-ms}
Consider a (finite, $n$-player) multi-objective normal-form game where players are optimising for the scalarised expected returns criterion. If each player $i$ has a quasiconvex utility function $u_i$, there must always exist an action $a^\ast_i$ that is a best response to the strategy profile of all other players $s_{-i}$ such that for all alternative strategies $s_i \in S_i$,

\begin{equation*}
    u_i(\mathbb{E}\left[\bm{p}_i (a^\ast_i, s_{-i})\right]) \geq u_i(\mathbb{E}\left[\bm{p}_i(s_i, s_{-i})\right]).
\end{equation*}

\end{lemma}

\begin{proof}

A mixed strategy $s_i$ assigns a probability to each $a \in A_i$, such that

\begin{equation}
    \sum_{j=1}^{|A_i|} \alpha_j = 1 \text{ and } \alpha_j \geq 0.
\end{equation}

\noindent The expected reward vector for $s_i$ is further defined as,
\begin{equation}
    \alpha_1 \bm{p}_i(a_1) + \cdots + \alpha_k \bm{p}_i(a_k) = \alpha_1 \bm{x}_1 + \cdots + \alpha_k \bm{x}_k.
\end{equation}
Note that the expected reward vector is thus a convex combination of the payoff vector for each action.

We claim that when assuming only quasiconvex utility functions, a pure strategy exists which is a best response. In other words, when applying a quasiconvex function $u$ to a convex combination of vectors, this is bounded by the maximum of $u$ applied to a single vector. Formally,
\begin{equation}
u(\alpha_1 \bm{x}_1 + \cdots + \alpha_m \bm{x}_m) \leq \max_{k} u(\bm{x}_k).
\end{equation}

In fact, this is a known property of quasiconvex functions \cite{dragomir2012jensen}, which proves Lemma \ref{lemma:optimal-ps-ms}.
\end{proof}

\begin{lemma}[Jensen's inequality for strictly quasiconvex functions]
\label{lemma:jensen-strict-qqconv}
Let $f$ be a strictly quasiconvex function and the points $\bm{x}_1, \cdots, \bm{x}_n$ in their domain not all equal. Then,
\begin{equation*}
    f(\lambda_1 \bm{x}_1 + \cdots + \lambda_n \bm{x}_n) < \max_{k} f(\bm{x}_k) \quad \text{where  } \lambda_1, \cdots, \lambda_n \in (0, 1), \sum_{k=1}^n \lambda_k = 1.
\end{equation*}
\end{lemma}

\begin{proof}
We prove this by induction on the number of points. Note that it is a straightforward extension of Jensen's inequality for quasiconvex functions \cite{dragomir2012jensen}.

For n = 2, it follows directly from the definition of a strictly quasiconvex function. Let us assume that the inequality holds when $\bm{x}_1, \cdots, \bm{x}_{n-1}$ are distinct and prove the inequality for $\bm{x}_1, \cdots, \bm{x}_{n}$ all distinct. Let ${\lambda_1, \cdots, \lambda_n \in (0, 1)}$, ${\sum_{k=1}^n \lambda_k = 1}$ and

\begin{equation}
    \bm{y} = \sum_{k=1}^{n-1} \frac{\lambda_k}{1-\lambda_n}\bm{x}_k.
\end{equation}
We can discern two possibilities. First, when $\bm{y} = \bm{x}_n$, we have by induction:
\begin{align*}
    f(\lambda_1 \bm{x}_1 + \cdots + \lambda_n \bm{x}_n) & = f\left(\left(1-\lambda_n \right)\bm{y} + \lambda_n\bm{x}_n \right) \\
    & = f\left(\left(1-\lambda_n \right)\bm{x}_n + \lambda_n\bm{x}_n \right) \\
    & = f(\bm{x}_n) \\
    & = f(\bm{y}) \\
    & = f\left(\sum_{k=1}^{n-1} \frac{\lambda_k}{1-\lambda_n}\bm{x}_k \right) \\
    & < \max_{k \in 1, \cdots, n-1} f(\bm{x}_k) \\
    & = \max_{k} f(\bm{x}_k)
\end{align*}
Alternatively, when $\bm{y} \neq \bm{x}_n$. Then, by the definition of strict quasiconvexity and induction:
\begin{align*}
    f(\lambda_1 \bm{x}_1 + \cdots + \lambda_n \bm{x}_n) & = f\left(\left(1-\lambda_n \right)\bm{y} + \lambda_n\bm{x}_n \right) \\
    & < \max \{f(\bm{y}), f(\bm{x}_n) \} \\
    & = \max \left\{f\left(\sum_{k=1}^{n-1} \frac{\lambda_k}{1-\lambda_n}\bm{x}_k \right), f(\bm{x}_n) \right\} \\
    & \leq \max \left\{\max_{k \in 1, \cdots, n-1} f(\bm{x}_k), f(\bm{x}_n) \right\} \\
    & = \max_{k} f(\bm{x}_k)
\end{align*}

Observe that we have shown Lemma \ref{lemma:jensen-strict-qqconv} holds when all points are distinct. When some $\bm{x}_k$'s are equal, we can group the expression $\sum_{k=1}^{n} \lambda_k \bm{x}_k$ into $\sum_{k=1}^{m} \mu_k \bm{y}_k$ where all $\bm{y}_k$'s are distinct and $m < n$. In this case, the result holds due to induction.
\end{proof}

\begin{lemma}[Mixed strategies as a best response]
\label{lemma:equal-payoff}
Consider a (finite, $n$-player) multi-objective normal-form game where players are optimising for the scalarised expected returns criterion. If each player $i$ has a strictly quasiconvex utility function $u_i$, a mixed strategy $s_i$ can only be a best response when the actions that are played with non-zero probability have equal expected returns to the strategy profile of all other players $s_{-i}$.
\end{lemma}

\begin{proof}
From Lemma \ref{lemma:optimal-ps-ms}, we know that a pure strategy will always be a best response when employing a quasiconvex utility functions. Lemma \ref{lemma:jensen-strict-qqconv} shows that for strictly quasiconvex functions, a combination of distinct points results in a lower return than the maximum individual point. This implies that for strictly quasiconvex utility functions, a mixed strategy will always return a lower utility than the optimal pure strategy when not all expected payoffs are equal for actions that are played with probability greater than zero. Therefore, the only case where a mixed strategy may be optimal is when it combines actions that have equal expected payoffs.
\end{proof}

We can now contribute Theorem \ref{th:no-ne-quasiconvex} which states that Nash equilibria may fail to exist when restricting players to use only strictly convex utility functions. Note that this proves immediately that no Nash equilibrium can be guaranteed when using only convex or quasiconvex utility functions, as these, by definition, contain the set of strictly convex functions.

\begin{theorem}
\label{th:no-ne-quasiconvex}
Consider a (finite, $n$-player) multi-objective normal-form game where players are optimising for the scalarised expected returns criterion. If each player has a strictly convex utility function, a mixed strategy Nash equilibrium is not guaranteed to exist.
\end{theorem}

\begin{proof}
Consider the game in Figure \ref{fig:counterexample-qconv} and utility functions

\begin{equation}
    u_1(p_1, p_2) = u_2(p_1, p_2) = p_1^2 + p_2^2.
\end{equation}

\begin{figure}[h]
    \centering
    \begin{game}{2}{2}
                  & $A$        & $B$\\
        $A$       & $(2, 0); (1, 0)$       & $(1, 0); (0, 2)$\\
        $B$       & $(0, 1); (2, 0)$       & $(0, 2); (0, 1)$
    \end{game}
    \caption{A game for which one can select strictly convex utility functions that do not admit a Nash equilibrium.}
    \label{fig:counterexample-qconv}
\end{figure}

This utility function is strictly convex in $\mathbb{R}^2$, and therefore also (continuous) strictly quasiconvex in $\mathbb{R}^2$. From Lemma \ref{lemma:equal-payoff}, we know that with such utility functions a mixed strategy can only be a best response when all actions that are played with non-zero probability return the same payoff vector. Because there are only two actions, any mixed strategy would have equal expected payoff for both actions in response to the opponent's strategy. Observe that this is never possible for either player due to the imbalanced return vector. As such, we may restrict the set of possible Nash equilibria to pure strategy Nash equilibria. It is trivial to verify that no pure-strategy Nash equilibrium exists and therefore no Nash equilibrium exists.
\end{proof}

In (single-objective) NFGs, there exist efficient methods for retrieving Nash equilibria \cite{porter2008simple}. By providing sufficient guarantees for NE existence in MONFGs, and importantly what restrictions do \emph{not} suffice, we open up the possibility for computational methods in this setting as well. Furthermore, algorithms in multi-agent reinforcement learning often aim for agents to converge to a Nash equilibrium \cite{busoniu2008comprehensive,nowe2012game}. Theorem \ref{th:ne-quasiconcave} guarantees the existence of an NE given certain restrictions, thus opening the possibility for new learning algorithms to be developed with the same goal in mind. We discuss these ideas as possible directions for future work in Section \ref{sec:conclusion}.

\section{Equilibrium Relations Between Optimisation Criteria}
\label{sec:existence-ne}
In this section, we shift our focus from the existence of Nash equilibria to the relations between Nash equilibria under both optimisation criteria. Specifically, we explore whether the number of Nash equilibria in an MONFG under SER equals that of the trade-off game and if any Nash equilibria need necessarily be shared. We find in both cases that even when a Nash equilibrium exists, the results are negative and that there is no guaranteed relation between the two in general.

We first contribute a theorem stating the fact that in an MONFG, the number of NE under SER can be different from the number of NE under ESR, even when both have NE.

\begin{theorem}
\label{th:ineq-monfg}
Consider a (finite, $n$-player) multi-objective normal-form game with at least one Nash equilibrium under both optimisation criteria. The size of the sets of Nash equilibria under the scalarised expected returns criterion and under the expected scalarised returns criterion need not be equal.
\end{theorem}

\begin{proof}
Consider the utility functions 
\begin{equation}
\label{eq:theorem-utility}
u_1(p_1, p_2) = u_2(p_1, p_2) = 0.1 \cdot p_1 + \max\{0, p_1\}  \cdot \max\{0, p_2\}
\end{equation}
and the MONFG under SER shown in Figure \ref{fig:counter-ser} with its trade-off game under ESR in Figure \ref{fig:counter-esr}.

\begin{figure}[h]
    \begin{subfigure}[t]{0.45\linewidth}
            \centering
            \begin{game}{2}{2}
                  & $A$        & $B$\\
                $A$       & $(1, 0); (1, 0)$       & $(0, 1); (0, 1)$\\
                $B$       & $(0, 1); (0, 1)$       & $(-10, 0); (-10, 0)$
            \end{game}
            \caption{The multi-objective reward vectors.}
            \label{fig:counter-ser}
    \end{subfigure}%
    \qquad
    \begin{subfigure}[t]{0.45\linewidth}
            \centering
            \begin{game}{2}{2}
                  & $A$        & $B$\\
                $A$       & \nashhighlight{$0.1; 0.1$}       & $0; 0$\\
                $B$       & $0; 0$       & $-1; -1$
            \end{game}
            \caption{The ESR utilities.}
            \label{fig:counter-esr}
    \end{subfigure}%
    \caption{An MONFG and its trade-off game from the utility functions shown in Equation \ref{eq:theorem-utility}. The MONFG shows by construction the two properties in Theorem \ref{th:ineq-monfg} and \ref{th:sharing-ne}. The highlighted cell is a pure strategy Nash equilibrium.}
    \label{fig:counter-ser-esr}
\end{figure}

To begin, let us show the NE in the MONFG under ESR. We do this by first applying the utility function for each agent -- which in this case happens to be the same -- directly to the payoff vectors in the MONFG. The resulting trade-off game can be seen in Figure \ref{fig:counter-esr}. We then observe that only the pure strategy profile ${(A, A)}$ results in utilities above $0$ for both agents. As such, there is no incentive for agents to play a mixed strategy when the other agent plays $A$ at least part of the time, leading to the pure strategy NE of ${(A, A)}$. Additionally, ${(B, B)}$ is not a NE, as there is an incentive for either agent to deviate to $A$, when the other plays $B$. This then again leads both agents to adapt their strategies to the NE of ${(A, A)}$, making it the only NE of the MONFG under ESR.

Next, we discuss the NE for the MONFG under SER (\ref{fig:counter-ser}). First note that the pure strategy NE of ${(A, A)}$ under ESR is not a NE under SER. To see this, observe that when one agent plays $A$ deterministically, the best response is found by maximising the univariate function 
\begin{equation}
\begin{split}
    \left(1 \cdot x + 0 \cdot \left(1 - x\right)\right) \cdot 0.1 \\ 
    + & \max\left\{0, \left(1 \cdot x + 0 \cdot \left(1 - x\right)\right)\right\} \\ 
    \cdot & \max\left\{0, \left(0 \cdot x + 1 \cdot \left(1 - x\right)\right)\right\} \quad \text{ with } 0 \leq x \leq 1
\end{split}
\end{equation}
with $x$ representing the probability of playing action $A$. Because strategies should sum to $1$, the probability of playing $B$ is equal to $(1-x)$. The solution to this maximisation is a mixed strategy with probability $\frac{11}{20}$ for action $A$ and probability $\frac{9}{20}$ for action $B$. This results in an expected return of $(\frac{11}{20}, \frac{9}{20})$ and a utility of ${0.1\cdot\frac{11}{20} + \frac{11}{20}\cdot\frac{9}{20}=0.3025}$ for both agents. In fact, this constitutes a NE under SER for this game, as no agent has an incentive to deviate from this strategy. A second NE occurs when the agents switch strategies, resulting in the same payoffs.

We can also show that the pure strategy ${(B, B)}$ is not a NE, as this can be improved upon by either agent deterministically playing A. As such, the MONFG in Figure \ref{fig:counter-ser-esr} has at least two mixed strategy NE under SER and no pure strategy NE. 

In this MONFG, both the game under SER and ESR have NE. However, we can see that they have a different number of NE, proving Theorem \ref{th:ineq-monfg}. 
\end{proof}

Given that the size of the sets of NE need not be equal, it is interesting to see whether any NE themselves need to be shared. Here too we are able to show that no such relation exists in general. We formalise this in Theorem \ref{th:sharing-ne}.

\begin{theorem}
\label{th:sharing-ne}
Consider a (finite, $n$-player) multi-objective normal-form game with at least one Nash equilibrium under both optimisation criteria. The set of Nash equilibria under the scalarised expected returns criterion and the set of Nash equilibria under the expected scalarised returns criterion may be disjoint.
\end{theorem}

\begin{proof}
Consider again the game in Figure \ref{fig:counter-ser-esr} and utility functions in Equation \ref{eq:theorem-utility}. We already highlighted the only Nash equilibrium for the MONFG under ESR, namely $(A, A)$. Moreover, we observe that this joint strategy is not a Nash equilibrium under SER as there is an incentive for either agent to deviate to playing the mixed strategy $(\frac{11}{20}, \frac{9}{20})$. As the MONFG has no other NE under ESR, no Nash equilibrium is shared in this construction.
\end{proof}

\section{Pure Strategy Nash Equilibria in MONFGs}
\label{sec:psne}
In the previous section, Theorem \ref{th:sharing-ne} stated that Nash equilibria need not be shared between optimisation criteria. A natural follow up question is under what circumstances this does occur. We expand on this issue and first show that a Pure Strategy Nash Equilibrium (PSNE) under SER must always be a PSNE under ESR as well. Furthermore, we demonstrate that the inverse does not hold by providing a counter example. However, we prove that adding the assumption that all utility functions in the MONFG are quasiconvex does ensure that PSNE under ESR are also a PSNE under SER. Finally, as a direct result of the new theorems, we show that the set of PSNE under ESR and SER are the same when assuming only quasiconvex utility functions. 

In order to show that a PSNE under SER must necessarily be a PSNE under ESR, we introduce Lemma \ref{lemma:utility}. This lemma states that the utility of a pure strategy profile under SER is the same as the utility of that pure strategy profile under ESR.

\begin{lemma}[Utility of a pure strategy]
\label{lemma:utility}
Given a pure strategy profile in a (finite, $n$-player) multi-objective normal-form game, the scalarised expected returns will always equal the expected scalarised returns. 
\end{lemma}

\begin{proof}
Consider a pure strategy profile $s$. We know that the observed payoff vector $\bm{p}(s)$ for this strategy will always be the same, even for multiple executions of the same strategy, as there is no randomisation over the actions. Because the expectation of a constant is equal to that constant we may state,
\begin{equation}
    \mathbb{E}[\bm{p}(s)] = \bm{p}(s)
\end{equation}

and given a utility function $u$, the expected utility also equals the received utility by the same reasoning,
\begin{equation}
\mathbb{E}[u(\bm{p}(s))] = u(\bm{p}(s)).
\end{equation}
We can thus say that for a pure strategy profile, the utility of a payoff under SER equals the utility under ESR:
\begin{equation}
u\left(\mathbb{E}\left[\bm{p}\left(s\right)\right]\right) = u\left(\bm{p}\left(s\right)\right) = \mathbb{E}\left[u\left(\bm{p}\left(s\right)\right)\right].
\end{equation}
\end{proof}

We note the importance of deterministic payoffs in MONFGs for Lemma \ref{lemma:utility}. If stochastic payoffs were allowed, the lemma, and by extension all other results in this section, would not hold. Given this lemma, we now define the first theorem of this section which states that a PSNE under SER, must always be a PSNE under ESR as well. In other words, the pure strategy Nash equilibria found in an MONFG under SER, must also be Nash equilibria in the trade-off game.

\begin{theorem}
\label{th:pure-strat-ser-esr}
Consider a (finite, $n$-player) multi-objective normal-form game with a pure strategy Nash equilibrium under the scalarised expected returns criterion. This joint pure strategy must necessarily also be a Nash equilibrium under the expected scalarised returns criterion.
\end{theorem}

\begin{proof}
Given a pure strategy Nash equilibrium under SER $s^{NE}$, we can say that:
\begin{align*}
& u_i\left(\mathbb{E}\left[ \bm{p}_i \left(s^{NE}_i, s^{NE}_{-i} \right)\right]\right) \geq u_i\left(\mathbb{E}\left[ \bm{p}_i \left(s_i, s^{NE}_{-i} \right)\right]\right) \\ 
\iff & u_i\left(\bm{p}_i \left(s^{NE}_i, s^{NE}_{-i} \right)\right) \geq u_i\left(\mathbb{E}\left[ \bm{p}_i \left(s_i, s^{NE}_{-i} \right)\right]\right) \\
\implies & \forall a \in A_i: u_i\left(\bm{p}_i \left(s^{NE}_i, s^{NE}_{-i} \right)\right) \geq u_i\left(\bm{p}_i\left(a, s_{-i}^{NE}\right)\right) \\
\iff & u_i\left(\bm{p}_i \left(s^{NE}_i, s^{NE}_{-i} \right)\right) \geq \max_{\bm{\alpha}} \sum_{a \in A_i} \alpha_a u_i\left(\bm{p}_i\left(a, s^{NE}_{-i}\right)\right) \\ 
& \text{ where } \forall a \in A_i: \alpha_a \geq 0,\sum_{a \in A_i} \alpha_a=1\\
\iff & \mathbb{E}\left[u_i\left(\bm{p}_i \left(s^{NE}_i, s^{NE}_{-i} \right)\right)\right] \geq \mathbb{E}\left[ u_i\left(\bm{p}_i \left(s_i, s^{NE}_{-i} \right)\right)\right] \\
\iff & \text{A pure strategy Nash equilibrium under ESR}
\end{align*} 
\end{proof}

The proof starts with the general definition of a pure strategy Nash equilibrium under SER and removes the expected values where possible. We then state that if the pure strategy profile is an NE, it must necessarily also be equal to or better than unilaterally playing another pure strategy. This leads us to state that the utility of the pure strategy NE is greater or equal to the optimal convex combination of the utilities of the other pure strategies. This is because such a combination of scalars is bounded by the maximum of these scalars. Next, we can freely introduce the expected value again in the left hand side of the inequality and rewrite the right hand side such that it now reflects the expected scalarised returns. This final inequality is also the definition of a Nash equilibrium under ESR. Given this positive result, it is alluring to believe that the inverse, so going from ESR to SER, would also hold. However, this is not the case as we can only guarantee that the utility of a pure strategy profile is greater or equal to the optimal stochastic mixture of scalar utilities. We can not guarantee that it is better than the utility of the optimal stochastic mixture of reward vectors.

\begin{theorem}
\label{th:pure-strat-esr-ser}
Consider a (finite, $n$-player) multi-objective normal-form game with a pure strategy Nash equilibrium under the expected scalarised returns criterion. This joint pure strategy need not be a Nash equilibrium under the scalarised expected returns criterion.
\end{theorem}

\begin{proof}
We show this theorem by using the same MONFG and utility functions as presented in the proof for Theorem \ref{th:ineq-monfg} and \ref{th:sharing-ne}. Recall that in this game, the pure strategy profile $(A, A)$ was a PSNE under ESR. This joint strategy was no PSNE under SER, as the best response to a player opting for $A$ was to play $A$ with probability $\frac{11}{20}$ and $B$ with probability $\frac{9}{20}$. 
\end{proof}
 
We add that an additional assumption can be made to remedy this negative result. Concretely, by making the assumption that all utility functions used by the players in the game are quasiconvex, we now show that a PSNE under ESR must also be a PSNE under SER. Note that contrary to Theorem \ref{th:ne-quasiconcave}, we do not need to restrict our utility functions to be continuous. 

\begin{theorem}
\label{th:pure-strat-esr-ser-qconvex}
Consider a (finite, $n$-player) multi-objective normal-form game with a pure strategy Nash equilibrium under the expected scalarised returns criterion. If all players have a quasiconvex utility function, this joint pure strategy must necessarily also be a Nash equilibrium under the scalarised expected returns criterion.
\end{theorem}

We first highlight that the construction created in the proof of Theorem \ref{th:pure-strat-esr-ser} uses a non-quasiconvex utility function. We observe this by applying the definition of quasiconvexity from Definition \ref{def:quasiconvex} to this utility function. If we take for example $\bm{x}_1 = (1, 0)$, $\bm{x}_2 = (0, 1)$ and $t=0.5$, then we get $f(0.5; 0.5) = 0.05 + 0.25 = 0.3$ for the left hand side and $\max\{f(1, 0), f(0, 1)\} = 0.05$ for the right hand side. It is clear then that this is not a quasiconvex utility function, as $0.3$ is larger than $0.05$. We now present the proof for Theorem \ref{th:pure-strat-esr-ser-qconvex}.


\begin{proof}
Given a pure strategy Nash equilibrium under ESR $s^{NE}$ and for each player $i$ a quasiconvex utility function $u_i$:

\begin{align*}
& \mathbb{E}\left[u_i\left(\bm{p}_i\left(s_i^{NE}, s_{-i}^{NE}\right)\right)\right] \geq \mathbb{E}\left[u_i\left(\bm{p}_i\left(s_i, s_{-i}^{NE}\right)\right)\right]\\
\wedge \quad & u_i\left(\bm{p}_i\left(s_i^{NE}, s_{-i}^{NE}\right)\right) = \max_{a \in A_i}u_i\left(\bm{p}_i\left(a, s^{NE}_{-i}\right)\right) \\
\wedge \quad & \forall \bm{\alpha}: \max_{a \in A_i} u_i\left(a, s^{NE}_{-i}\right) \geq u_i\left(\sum_{a \in A_i} \alpha_a \bm{p}_i\left(a, s^{NE}_{-i}\right)\right) \\
& \text{ where }  \forall a \in A_i: \alpha_a \geq 0,\sum_{a \in A_i} \alpha_a=1\\
\implies & \max_{a \in A_i} u_i\left(a, s^{NE}_{-i}\right) \geq u_i\left(\mathbb{E}\left[\bm{p}_i\left(s_i, s_{-i}^{NE}\right)\right]\right)\\
\implies & u_i\left(\bm{p}_i\left(s_i^{NE}, s_{-i}^{NE}\right)\right) = \max_{a \in A_i}u_i\left(\bm{p}_i\left(a, s^{NE}_{-i}\right)\right) \geq u_i\left(\mathbb{E}\left[\bm{p}_i\left(s_i, s_{-i}^{NE}\right)\right]\right) \\
\implies & u_i\left(\bm{p}_i\left(s_i^{NE}, s_{-i}^{NE}\right)\right) \geq u_i\left(\mathbb{E}\left[\bm{p}_i\left(s_i, s_{-i}^{NE}\right)\right]\right) \\
\implies & u_i\left(\mathbb{E}\left[\bm{p}_i\left(s_i^{NE}, s_{-i}^{NE}\right)\right]\right) \geq u_i\left(\mathbb{E}\left[\bm{p}_i\left(s_i, s_{-i}^{NE}\right)\right]\right) \\
\implies &\text{A pure strategy Nash equilibrium under SER}
\end{align*} 
\end{proof}

Here too, we present a textual walk-through of the proof. Initially, we introduce the definition of a Nash equilibrium under ESR. We then observe that because we are dealing with pure strategies, we may eliminate the expectation from the left hand side following Lemma \ref{lemma:utility}. In the right hand side, we state that the utility of this PSNE must equal the utility of the best-response pure strategy. This follows directly from the definition of a (pure strategy) Nash equilibrium. 

We also know that given the extension of Jensen's inequality to quasiconvex functions \cite{dragomir2012jensen}, any convex combination of payoff vectors is bounded by the maximum utility of one such vectors. This implies that the utility of the optimal action is greater than that of any alternative strategy under SER. In the next line, we reintroduce the left hand side from earlier, while keeping the inequality intact. We then remove the maximum, after which we add the expectation into the left hand side. By doing this, we have arrived at the definition of a NE under SER, proving Theorem \ref{th:pure-strat-esr-ser-qconvex}. 

From Theorem \ref{th:pure-strat-ser-esr} and \ref{th:pure-strat-esr-ser-qconvex}, we are able to present our final result of this section, observing that when assuming only quasiconvex utility functions for players in an MONFG, the set of PSNE under SER is equal to the set of PSNE under ESR. We state this formally in Theorem \ref{th:PSNE-esr-ser-qconvex-equal}. 

\begin{theorem}
\label{th:PSNE-esr-ser-qconvex-equal}
Consider a (finite, $n$-player) multi-objective normal-form game where each player has a quasiconvex utility function. The set of pure strategy Nash equilibria under the expected scalarised returns criterion is equal to the set of pure strategy Nash equilibria under the scalarised expected returns criterion.
\end{theorem}

\begin{proof}
This follows directly from Theorems \ref{th:pure-strat-ser-esr} and \ref{th:pure-strat-esr-ser-qconvex}. We know that any PSNE under SER is also a PSNE under ESR. Vice versa, we know that when all players have a quasiconvex utility function, any PSNE under ESR is also a PSNE under SER.  
\end{proof}



This theorem has important implications. Recall that an MONFG under the ESR criterion can always be reduced to an NFG \cite{radulescu2020utility}. We now know that in an MONFG where all players have a quasiconvex utility function, the set of PSNE under both criteria is equal. This indicates that if all utility functions are quasiconvex, we can retrieve pure strategy Nash equilibria for any multi-objective game by looking at the trade-off game. This equivalence lends itself to a novel computational method for retrieving all pure strategy Nash equilibria in multi-objective games. We show this approach in Section \ref{sec:psne-algorithm}. 

In addition, Theorem \ref{th:PSNE-esr-ser-qconvex-equal} allows us to provide an additional existence guarantee for Nash equilibria in MONFGs. Specifically, for any class of NFGs for which the set of PSNE is guaranteed to be non-empty, we can subsequently guarantee a PSNE for any MONFG whose scalarisation with quasiconvex utility functions leads to a game in that original class. Such guarantees for NFGs exist \cite{rosenthal1973class}.

\section{Blended Settings in MONFGs}
\label{sec:blended}
In this section, we briefly move away from comparing the Nash equilibria in MONFGs under both criteria and discuss what happens when allowing for a heterogeneous set of players optimising for different criteria, i.e. some players optimising for the ESR criterion, while others are optimising for SER. We refer to this setting as a \emph{blended setting} and to players in such a setting as a \emph{blended set of players}. Lastly, we refer to the \emph{distribution of players} to denote the proportions of players optimising for either SER or ESR.

Such types of settings are interesting to study, as they are sure to arise in the real world (see Example \ref{exmp:blended}). While this setting is surely worth attention, it is left almost completely unexplored \cite{radulescu2020multiobjective}. We first contribute a novel definition for a Nash equilibrium in a blended setting. Informally, this definition states that no agent can improve on their individual criterion by unilaterally deviating from the joint strategy.

\begin{definition}[Nash equilibrium in a blended setting]
\label{def:MOMA-NE-blended}
A joint strategy $s^{NE}$ leads to a Nash equilibrium in a blended setting if for each agent optimising for the ESR criterion $i \in I \subseteq N$ and each agent optimising for the SER criterion $j \in J \subseteq N$ so that $I \bigcup J = N$ and for all alternative strategies $s_i \in S_i$ and $s_j \in S_j$:
\[
\forall i \in I, \quad \mathbb{E}\left[u_i\left(\mathbf{p}_i\left(s_{i}^{NE}, s_{-i}^{NE}\right)\right] \right) \geq  \mathbb{E}\left[u_i\left(\mathbf{p}_i\left(s_i, s_{-i}^{NE}\right) \right)\right]
\]
and 
\[
\forall j \in J, \quad u_j\left(\mathbb{E}\left[ \bm{p}_j \left(s^{NE}_j, s^{NE}_{-j} \right)\right]\right) \geq u_j\left(\mathbb{E}\left[ \bm{p}_j \left(s_j, s^{NE}_{-j} \right)\right]\right) 
\]
\noindent i.e. $s^{NE}$ is a Nash equilibrium in a blended setting if no agent optimising for the ESR criterion can increase the \emph{expected utility of its payoffs} and no agent optimising for the SER criterion can increase the \emph{utility of its expected payoffs} by deviating unilaterally from $s^{NE}$.
\end{definition}

In a best case scenario, it would be possible to calculate or learn Nash equilibria for any given MONFG with a blended set of players \emph{without knowing the distribution of these players}. Specifically, this would mean that the game is robust to any changes in the player distribution and shows some static characteristics. Luckily, we are able to use the theorems provided in Section \ref{sec:psne} to derive several such properties. We first observe that if a PSNE exists in an MONFG where every player is optimising for the SER criterion, it is also guaranteed to be a PSNE in any blended setting. Formally: 

\begin{corollary}
\label{co:ser-blended-psne}
Consider a (finite, $n$-player) multi-objective normal-form game where players are optimising for the scalarised expected returns criterion. A pure strategy Nash equilibrium in this setting is also a Nash equilibrium in any blended setting.
\end{corollary}

\begin{proof}
The proof follows directly from Theorem \ref{th:pure-strat-ser-esr}. Starting with this theorem, we know that any pure strategy that is a Nash equilibrium under SER is also a NE under ESR. Given an MONFG with a PSNE under SER, players optimising for this criterion in a blended setting have no incentive to deviate. As it is also a PSNE under ESR, the players optimising for this criterion do not have an incentive to deviate either. As such, this strategy presents a PSNE in any blended setting.
\end{proof}

Note that following the same logic, we also guarantee that any PSNE in a blended setting is a PSNE in the trade-off game. Observe that PSNE are not altered by players shifting from the SER criterion to ESR, as it already was a PSNE under ESR. In addition to this we can use the theorems from Section \ref{sec:psne} to show that if a PSNE exists in a game scalarised using quasiconvex utility functions, it must necessarily also be a PSNE in any blended setting. We formalise this in Corollary \ref{co:blended-psne}.

\begin{corollary}
\label{co:blended-psne}
Consider a (finite, $n$-player) multi-objective normal-form game where each player has a quasiconvex utility function. A pure strategy Nash equilibrium under the expected scalarised returns criterion is also a Nash equilibrium in any blended setting.
\end{corollary}

\begin{proof}
From Theorem \ref{th:PSNE-esr-ser-qconvex-equal}, we know that when assuming only quasiconvex utility functions, a PSNE under ESR is necessarily also a PSNE under SER. Thus, by the same logic used in the previous proof, no player can improve on their particular criterion by deviating. As such, we have a pure strategy Nash equilibrium in any blended setting.
\end{proof}

We can use this corollary in a final theorem. Intuitively, we observe that the set of PSNE in any blended setting with players employing quasiconvex utility functions is equal to the set of PSNE in the trade-off game. Formally:

\begin{theorem}
\label{th:blended-psne-equal}
Consider a (finite, $n$-player) multi-objective normal-form game where each player has a quasiconvex utility function. The set of pure strategy Nash equilibria in the trade-off game is equal to the set of pure strategy Nash equilibria in any blended setting.
\end{theorem}

\begin{proof}
This follows directly from Theorem \ref{th:PSNE-esr-ser-qconvex-equal} and Corollary \ref{co:blended-psne}. First, Corollary \ref{co:blended-psne} guarantees that a PSNE under SER is also a PSNE in any blended setting. By extension, this shows that a PSNE in a blended setting is also a PSNE in the trade-off game. Theorem \ref{th:PSNE-esr-ser-qconvex-equal} subsequently guarantees that these PSNE are also PSNE under SER.
\end{proof}


%

The results presented in this section have several important implications. Specifically, we now guarantee that in certain situations we can determine pure strategy Nash equilibria for blended settings \emph{a priori}, without even knowing the distribution of players optimising for either criterion. In particular, Theorem \ref{th:blended-psne-equal} shows that when players are employing quasiconvex utility functions, the set of PSNE in a blended setting corresponds to that of the trade-off game. In the following section, we use this property to derive a general algorithm for calculating PSNE in MONFGs under ESR, SER or indeed any blended setting given quasiconvex utility functions.

\section{Algorithmic Implications}
\label{sec:psne-algorithm}
As a final contribution, we propose a novel algorithm for computing all pure strategy Nash equilibria in finite $n$-player MONFGs, given quasiconvex utility functions. From this point onward, we will only consider such utility functions when discussing the algorithm. We highlight that the approach laid out in this section will be shown to operate correctly for any optimisation criterion or blended setting. We provide a pseudocode implementation in Algorithm \ref{alg:PSNE}. A functioning implementation can be found at \url{https://github.com/wilrop/moqups}.

The algorithm is designed in two sequential components. First, we use the fact that when optimising for the ESR criterion in an MONFG, the game can be trivially reduced to a single-objective NFG \cite{radulescu2020utility}. Specifically, given an MONFG $G=(N, \mathcal{A}, \bm{p})$, we can construct a single-objective NFG $G'=(N, \mathcal{A}, f)$ where $N$ and $\mathcal{A}$ are the same. The payoff functions for the NFG are defined as the composition between player $i$'s utility function $u_i$ and their vectorial payoff function $\bm{p}_i$. This construction is defined in the function \texttt{reduce\_monfg} from line 1 to 6.

Second, from this reduction it follows that we can use algorithms from single-objective NFGs to retrieve all PSNE in the original MONFG. If one or more of such a PSNE exists, this would further imply that the same strategy is a PSNE in the MONFG under ESR. By Theorem \ref{th:pure-strat-esr-ser-qconvex} we then also have a PSNE under SER and by Corollary \ref{co:blended-psne} we also have a PSNE in any blended setting. Lastly, we know by Theorems \ref{th:PSNE-esr-ser-qconvex-equal} and \ref{th:blended-psne-equal} that the set of PSNE retrieved by this method necessarily contains all PSNE in the MONFG, irrespective of the player distribution. We implement a straightforward method for computing all PSNE in an NFG in the function \texttt{compute\_all\_PSNE} from line 8 to 16. 

\begin{algorithm}[h]
    \caption{Computing all PSNE in an MONFG}\label{alg:PSNE}
    \textbf{Input:} an MONFG $G=(N, \mathcal{A}, \bm{p})$ and quasiconvex utility functions $u=(u_1, \cdots, u_n)$
    \begin{algorithmic}[1]
        \Function{reduce\_monfg}{monfg, u}
            \State $N, \mathcal{A}, \bm{p} \gets$ monfg
            \State $u_1, \cdots, u_n \gets$ u
            \State $f \gets (u_1 \circ \bm{p}_1, \cdots, u_n \circ \bm{p}_n)$
            \State $G'\gets (N, \mathcal{A}, f)$  \Comment{An induced normal-form game}
            \State \Return $G'$
        \EndFunction
        \Function{compute\_all\_PSNE}{nfg}
            \State $S = \emptyset$
            \For{PS \textbf{in} nfg}  \Comment{Loop over all pure strategies}
                \If{PS is a PSNE}  \Comment{If it is a PSNE add it to the set}
                    \State {$S \gets S \cup \{\text{PS}\}$}
                \EndIf
            \EndFor
            \State \Return $S$
        \EndFunction
    \State nfg $\gets$ \textproc{reduce\_monfg}($G$, $u$)
    \State PSNE $\gets$ \textproc{compute\_all\_PSNE}(nfg)
    \end{algorithmic}
\end{algorithm}

The algorithm sequentially goes through these two parts. We first call the reduction in line 17, after which we retrieve all PSNE in line 18. One important benefit of this approach is that the complexity does not increase when including more objectives. Instead, only the runtime is increased when performing the initial scalarisation for more objectives. This makes a game with $d$ objectives effectively as hard to solve as a game with 2 objectives.

For simplicity, we employ a naive support enumeration approach to compute all pure strategy Nash equilibria in the single-objective NFG. There also exist faster algorithms with improved computational complexity \cite{corley2020regretbased}. Furthermore, when the induced single-objective game has a specific form, one could use a tailored algorithm for computing all pure strategy Nash equilibria. Lastly, parallelism could be employed to further speed up the required computations.

We note the fact that Algorithm \ref{alg:PSNE} can be trivially modified to compute a single sample pure strategy Nash equilibrium, rather than all PSNE. This can be achieved by first performing the initial reduction to a single-objective NFG and subsequently applying a method for computing a sample PSNE rather than all PSNE. 

\section{Related Work}
\label{sec:related-work}
Multi-objective normal-form games were introduced by Blackwell \cite{blackwell1954analog} and have since been considered with many different approaches. We highlight below several important works and frame their contributions in terms of similarities and differences to the utility-based approach we employ in this article.

Early studies on multi-objective games, often referred to as multi-criteria games, largely focused on arguing its importance and extending relevant solution concepts from single-objective game theory to this setting \cite{blackwell1954analog,shapley1959equilibrium,zeleny1975games}. Most works considered the case in which agents do not know their utility function, and thus define utility function agnostic equilibria. Shapley and Rigby \cite{shapley1959equilibrium} introduce the now widely used concept of Pareto Nash equilibria. They extend and characterise the set of mixed-strategy agnostic Pareto-Nash equilibria for multi-objective two-person zero-sum games for linear utility functions. We stress here the apparent disconnect between the proposed solution concept and presented characterisation. Specifically, mixed strategy Pareto Nash equilibria are defined as strategies that result in non-dominated vectors. However, by their definition this implies that we are considering \emph{expected returns} thus leading us to the SER criterion in the utility-based approach. When subsequently scalarising the game, utility functions are conveniently restricted to be linear only, as this allows for discussion of these points in the equivalent single-objective game. As such, they implicitly shift to the \emph{expected utility} or ESR in the utility-based framework. Note that this final scalarisation is only mathematically sound when restricting to linear utility functions \cite{radulescu2020utility}.

We highlight that such conflicting approaches are regularly seen in multi-objective game theory. On the one hand, often a utility-agnostic approach is assumed and expected vectors (i.e. SER) are considered \cite{shapley1959equilibrium,voorneveld2000ideal,ismaili2018existence}. On the other hand, linear utility functions are assumed in order to facilitate a reduction to an equivalent single-objective game (i.e. ESR) and allow claims to be made about the resulting equilibria \cite{corley1985games,lozovanu2005multiobjective,zeleny1975games}. Note that under such linear utility functions, ESR and SER are in fact equivalent \cite{radulescu2020utility}.

There has also been work on formulating algorithms for finding Pareto-Nash equilibria in multi-objective non-cooperative games. Lozovanu et al. \cite{lozovanu2005multiobjective} propose a method that computes the trade-off game (i.e., implicitly assume the ESR criterion) for every linear utility function for which the weights sum to one and subsequently find its NE. In addition, recent work by Ismaili \cite{ismaili2018existence} highlights the specific properties of pure strategy Pareto-Nash equilibria when considering expected payoff vectors and provides an algorithm for computing these equilibria.

In this work we assume a utility-based perspective. This approach is introduced by Roijers et al. \cite{roijers2013survey} who subsequently make the distinction between the ESR and SER criterion explicit and highlight their differences. While their work focused mostly on the single-agent case, a recent survey assumes this approach in multi-objective multi-agent settings \cite{radulescu2020multiobjective}. They further offer a taxonomy of such decision making problems on the basis of payoffs, utility and the type of desired outcomes. 

Conforming to this taxonomy, in this paper we focus on individual utility, i.e., even if agents receive the same payoff vector they may value this payoff vector differently. Furthermore, we assume that no social welfare mechanism is employed, and that we are looking for stable outcomes in settings with self-interested agents. This approach is also assumed by R\u{a}dulescu et al. \cite{radulescu2020utility}, who provide several foundational results that we build upon in this article. First, they show that an MONFG under the ESR criterion can always be reduced to an NFG and that when using linear utility functions, an MONFG can always be scalarised a priori. Furthermore, they formally show that under SER no NE need necessarily exist. 

As previously mentioned, throughout most earlier work the (implicit) assumption is made that utility functions are linear. Nonetheless, it is a well-known fact that utility functions can be highly non-linear. This has been noted in a variety of scenarios ranging from queueing games \cite{breinbjerg2017equilibrium} to travel choice behaviour \cite{koppelman1981nonlinear}. Bergstresser and Yu \cite{bergstresser1977domination} bring up the idea that utility functions could also be non-linear in multi-objective normal-form games. However, in their practical analysis, they only consider linear utility functions and apply the ESR criterion to obtain the resulting trade-off game and corresponding solution points. Recent work in MONFGs often explores reinforcement learning techniques to find optimal solutions for agents in such settings when operating under non-linear utility functions \cite{radulescu2022opponent,ropke2021communication}. Lastly, in settings where the utility functions are not easily expressed, it can prove interesting to explore different elicitation strategies. One approach that has proven successful examined the use of Gaussian processes to elicit user preferences by asking the user to rank items or perform simple pairwise comparisons \cite{roijers2021interactive,zintgraf2018ordered}.

\section{Conclusion and Future Work}
\label{sec:conclusion}
We explored the theoretical foundations of normal-form games with vectorial payoffs, also referred to as multi-objective normal-form games. To this end, we assumed a utility-based approach that guarantees the existence of a utility function that can be used to scalarise a payoff vector \cite{hayes2022practical,roijers2017multi}. A known complicating factor is that such a scalarisation can occur at two stages, specifically when considering the utility of mixed strategies. We refer to these approaches as optimisation criteria, as agents are rational actors that wish to optimise for their particular criterion. The first option is to scalarise the payoff vectors for each joint action in the payoff matrix directly and calculate the expected utility of mixed strategies, also referred to as the expected scalarised returns criterion (ESR). On the other hand, we can also calculate the expected payoff vector with regards to the mixed strategy and calculate the utility of this expected vector. This criterion is also known as the scalarised expected returns criterion (SER). Earlier work showed that in this latter approach, no Nash equilibrium needs to exist \cite{radulescu2020utility}. 

We first prove that the existence of a mixed strategy Nash equilibrium is guaranteed in the setting of MONFGs under SER when we restrict agents to employ only continuous quasiconcave utility functions. Second, we prove that under strictly convex utility functions, and therefore quasiconvex utility functions as well, Nash equilibria are not guaranteed to exist. Given that agents in multi-agent systems are often assumed or incentivised to play according to Nash equilibria, such guarantees will prove important in the development of future algorithms.

Next, we explored whether there exists a relationship in either the number of Nash equilibria or the equilibria themselves in MONFGs under both optimisation criteria. We found that even when an equilibrium exists under both criteria, the number of Nash equilibria can differ. In addition, no Nash equilibria need necessarily be shared.

These negative results led us to restrict our focus to only pure strategy Nash equilibria (PSNE). We show that we can guarantee that the set of pure strategy Nash equilibria under SER equals that under ESR when assuming quasiconvex utility functions. When disregarding this assumption, we are only able to guarantee that a PSNE under SER is also a PSNE under ESR. Due to the generality of this result, we further investigated blended settings in which agents are allowed to optimise for different criteria. We contributed a novel definition of a Nash equilibrium for this setting and subsequently showed that when assuming only quasiconvex utility functions, the set of PSNE in any blended setting is equivalent to that of the trade-off game. 

Our last contribution added a novel algorithm for computing all PSNE in any game with vectorial payoffs and quasiconvex utility functions, irrespective of whether agents are optimising for ESR, SER or any blended settings. This algorithm first scalarises the input MONFG and subsequently calculates the PSNE in this game. Due to the previous theorems, we can guarantee the correctness of this approach.

For future work, we aim to focus on two possible directions. First, we propose to look further into the theoretical foundations of MONFGs. Given that quasiconcave utility functions guarantee Nash equilibria, it is interesting to explore what other assumptions on utility functions give rise to NE. In addition, it can prove interesting to study specific restrictions on payoff structures such as zero-sum games to find whether the assumptions on utility functions can be relaxed. Furthermore, we already addressed that blended settings are left almost completely unexplored. Due to the practical relevance of such settings, we aim to provide stronger theoretical guarantees in these scenarios as well. 

The second direction for future work is the algorithmic aspect. We already contribute an initial algorithm for computing all PSNE in an MONFG when assuming quasiconvex utility functions. We believe that designing an algorithm that is able to do this for any type of utility function or even for mixed-strategy NE in general will prove a worthwhile pursuit. In single-objective NFGs, such approaches already exist. Furthermore, more efficient methods can be applied on subclasses of NFGs that contain specific properties such as a zero-sum payoff structure. Studying whether such methods apply in the setting of MONFGs can also prove important.

\section*{Acknowledgments}
The first author is supported by the Research Foundation – Flanders (FWO), grant number 1197622N. This research was supported by funding from the Flemish Government under the ``Onderzoeksprogramma Artifici\"{e}le Intelligentie (AI) Vlaanderen'' program.

\bibliographystyle{unsrtnat}
\bibliography{bibliography}  






\end{document}